\documentclass[12 pt, draft, onecolumn]{IEEEtran}

\usepackage[cmex10]{amsmath}
\usepackage{amssymb}
\usepackage{mathrsfs}
\newtheorem{lemma}{Lemma}
\newcommand{\expectation}{\ensuremath{\mathbb{E}}}
\newcommand{\Expt}{\expectation}

\begin{document}

\title{Submartingale Property of $E_0$ Under The Polarization Transformations}
\author{\IEEEauthorblockN{Mine Alsan, Emre Telatar}\\
\small\IEEEauthorblockA{Information Theory Laboratory\\
Ecole Polytechnique F\' ed\' erale de Lausanne\\
CH-1015 Lausanne, Switzerland\\
Email: mine.alsan@epfl.ch, emre.telatar@epfl.ch}
\normalsize}
\maketitle
\pagestyle{empty}
\thispagestyle{empty}

Given a binary input channel $W$, let $E_0(\rho,W)$ denote ``Gallager's $E_0$'' \cite[p.~138]{578869} evaluated for the uniform input distribution:
\begin{equation} \label{eq:E0}
E_0(\rho, W) = -\log \displaystyle\sum_{y\in \mathcal{Y}} \left[ \frac{1}{2} W(y\mid 0)^{\frac{1}{1+\rho}} + \frac{1}{2} W(y\mid 1)^{\frac{1}{1+\rho}} \right]^{1+\rho}.   
\end{equation} 

In this note we prove that the following relation
\begin{equation}\label{eq::EO_submartingale}
E_0(\rho, W^{-}) + E_0(\rho, W^{+}) \geq 2 E_0(\rho, W)
\end{equation} 
holds for any binary input discrete memoryless channel (B-DMC) $W$, and $\rho\geq 0$.
The channels $W^{-}$, and $W^{+}$ denote the synthesized channels after the application of the one step polarization transformations defined by Ar{\i}kan~\cite{1669570}. 
Their transition probabilities are given by:
\begin{align} 
 \label{align:trans1} &W^{-}(y_1 y_2 \mid u_1) = \displaystyle\sum_{u_2\in\{0, 1\}} \frac{1}{2} W(y_1 \mid u_1 \oplus u_2) W(y_2 \mid u_2) \\
 \label{align:trans2} &W^{+}(y_1 y_2 u_1\mid u_2) = \frac{1}{2} W(y_1 \mid u_1 \oplus u_2) W(y_2 \mid u_2). 
\end{align}

The special case of the relation above with $\rho=1$ was proved in~\cite{1669570}.  Another special case of the relation,
by first dividing by $\rho$ and taking the limit as $\rho$ tends to zero is also shown in~\cite{1669570} as a consequence
of the chain rule for mutual information.  We simply provide the extension of these results to arbitrary, non-negative
values of $\rho$.

\begin{proof}
By Lemmas \ref{lem:basic}, \ref{lem:minus}, and \ref{lem:plus} proved in the Appendix, we know that 
\begin{align*}
 E_0(\rho, W) &= -\log \Expt[g(\rho, Z)]  \\
 E_0(\rho, W^{-}) &= -\log \Expt[g(\rho, Z_1 Z_2)]  \\
 E_0(\rho, W^{+}) &= -\log \Expt[h(\rho, Z_1, Z_2)] 
\end{align*}
where $Z, Z_1, Z_2$ are independent, identically distributed random variables taking values in the $[0, 1]$ interval, and
\begin{align*}
 g(\rho, z) &\triangleq \left(\frac{1}{2}(1+z)^{\frac{1}{1+\rho}} + \frac{1}{2}(1-z)^{\frac{1}{1+\rho}}\right)^{1+\rho} \\
 h(\rho, z_1, z_2) &\triangleq \frac{1}{2}(1 + z_1z_2)g\Bigl(\rho, \frac{z_1 + z_2}{1 + z_1z_2}\Bigr) + \frac{1}{2}(1 - z_1z_2)g\Bigl(\rho, \frac{z_1 - z_2}{1 - z_1z_2}\Bigr).
\end{align*}
By these identities, showing \eqref{eq::EO_submartingale} is equivalent to showing
\begin{equation*}
\Expt[g(\rho, Z_1)]\Expt[g(\rho, Z_2)] \geq \Expt[g(\rho, Z_1 Z_2)]\Expt[h(\rho, Z_1, Z_2)].
\end{equation*}
\\
The proof is carried in two steps. We first claim that the following inequality is satisfied:
\begin{equation}\label{eq::inequality}
g(\rho, z_1) g(\rho, z_2) \geq g(\rho, z_1z_2)h(\rho, z_1, z_2)
\end{equation}
for any $z_1, z_2\in[0, 1]$, and $\rho\geq 0$. \\

Taking the expectation of both sides in \eqref{eq::inequality} and noting the independence of $Z_1$ and $Z_2$ gives
\begin{equation}
\Expt\left[g(\rho,Z_1)\right] \Expt\left[g(\rho,Z_2)\right]=
\Expt\left[g(\rho, Z_1) g(\rho, Z_2) \right] \geq  \Expt\left[g(\rho, Z_1 Z_2)h(\rho, Z_1, Z_2) \right].
\end{equation}
By Lemma \ref{lem::f_decreasing} in the Appendix, the function $g(\rho, z_1 z_2)$ is non-increasing in $z_1$, and $z_2$ separately for any $\rho\geq 0$. 
Similarly, by Lemma \ref{lem::h_decreasing} in the Appendix the function 
$h(\rho, z_1, z_2)$ is also non-increasing in both $z_1$, and $z_2$ separately for any $\rho\geq 0$. The monotonicity properties are useful as they imply (see, e.g., \cite[Ch.~9, p.~446-447]{Stochastic_Proc_Ross}) that the random variables $g(\rho,Z_1Z_2)$ and $h(\rho,Z_1,Z_2)$ are positively correlated.  As a result
\begin{equation}
 \Expt\left[g(\rho, Z_1)\right] \Expt\left[g(\rho, Z_2) \right] \geq  \Expt\left[g(\rho, Z_1 Z_2)h(\rho, Z_1, Z_2) \right] \geq \Expt\left[g(\rho, Z_1 Z_2)\right] \Expt\left[h(\rho, Z_1, Z_2) \right],
\end{equation}
concluding the proof of the relation given in \eqref{eq::EO_submartingale}.\\

Now, we prove the claimed inequality in \eqref{eq::inequality}. For that purpose, we first apply the transformations 
\begin{equation*}
s = \displaystyle\frac{1}{1+\rho}, \quad t = \operatorname{arctanh}{z_1}, \quad w = \operatorname{arctanh}{z_2}, \quad k = \operatorname{arctanh}(z_1z_2) 
\end{equation*}
where $s\in[0, 1]$, and $t, w, k\in[0, \infty)$. Using these, we obtain
\begin{align}
 \label{eq::g_1} g\Bigl(\frac{1-s}{s}, \tanh(t)\Bigr) &= \displaystyle\frac{\cosh(s t)^{\frac{1}{s}}}{\cosh(t)} \\
 \label{eq::g_2} g\Bigl(\frac{1-s}{s}, \tanh(w)\Bigr) &= \displaystyle\frac{\cosh(s w)^{\frac{1}{s}}}{\cosh(w)}
\end{align}
and
\begin{align}
 \label{eq::g_3} g\Bigl(\frac{1-s}{s}, \tanh(k)\Bigr) &= \displaystyle\frac{\cosh(s k)^{\frac{1}{s}}}{\cosh(k)} \\
 \label{eq::h} h\Bigl(\frac{1-s}{s}, \tanh(t), \tanh(w)\Bigr) &= \displaystyle\frac{\cosh(s(t+w))^{\frac{1}{s}} + \cosh(s(t-w))^{\frac{1}{s}}}{2\cosh(t)\cosh(w)}.
\end{align}
We further define the transformations
\begin{equation*}
 a = t + w, \quad b = t - w 
\end{equation*}
such that $t = \displaystyle\frac{a+b}{2}$, and $w = \displaystyle\frac{a-b}{2}$ where $a \geq |b|$.
Then, the variable $k$ is given by
\begin{equation}
 k = \displaystyle\frac{1}{2}\log\left(\displaystyle\frac{\cosh(a)}{\cosh(b)}\right).
\end{equation}
Therefore, the expression in \eqref{eq::g_3} becomes
\begin{equation} \label{eq::g_4}
 g\Bigl(\frac{1-s}{s}, \tanh(k)\Bigr) = \displaystyle\frac{\left(\displaystyle\frac{\cosh(a)^s + \cosh(b)^s}{2}\right)^{\frac{1}{s}}}{\displaystyle\frac{\cosh(a) + \cosh(b)}{2}}.
\end{equation}

After a few manipulations on the product of the equations \eqref{eq::g_1}, and \eqref{eq::g_2}, one can check that the LHS of \eqref{eq::inequality} is given by
\begin{equation}
\displaystyle\frac{\left(\frac{\cosh(s a) + \cosh(s b)}{2}\right)^{\frac{1}{s}}}{\cosh(t)\cosh(w)}. 
\end{equation}

Similarly, using equations \eqref{eq::h}, and \eqref{eq::g_4}, the RHS of \eqref{eq::inequality} is given by
\begin{equation}
\displaystyle\frac{\left(\frac{\cosh(a)^s + \cosh(b)^s}{2}\right)^{\frac{1}{s}}}{\frac{\cosh(a) + \cosh(b)}{2}}\times\displaystyle\frac{\cosh(s a)^{\frac{1}{s}} + \cosh(s b)^{\frac{1}{s}}}{2\cosh(t)\cosh(w)}.  
\end{equation}

Therefore, we obtain that the inequality \eqref{eq::inequality} is equivalent to 
\begin{equation}
\displaystyle\frac{\left(1 + \left(\frac{\cosh(b s)^\frac{1}{s}}{\cosh(a s)^\frac{1}{s}}\right)^s\right)^{\frac{1}{s}}}{1 +  \frac{\cosh(b s)^{\frac{1}{s}}}{\cosh(a s)^{\frac{1}{s}}}} 
\geq  \displaystyle\frac{\left(1 + \left(\frac{\cosh(b)}{\cosh(a)}\right)^s\right)^{\frac{1}{s}}}{1 + \frac{\cosh(b)}{\cosh(a)}}. 
\end{equation}
Let $u = \frac{\cosh(b s)^{\frac{1}{s}}}{\cosh(a s)^{\frac{1}{s}}}$, and $v = \frac{\cosh(b)}{\cosh(a)}$. Then, by Lemma \ref{lem::f_decreasing} in the Appendix, whenever $a \geq b \geq 0$, we have $u \geq v$ since 
\begin{equation*}
f_s(b) = \frac{\cosh(b s)^{\frac{1}{s}}}{\cosh(b)} \geq \frac{\cosh(a s)^{\frac{1}{s}}}{\cosh(a)} = f_{s}(a).  
\end{equation*}
Moreover, we have $u \geq v$ whenever $a \geq |b|$ by symmetry of the function $f_s(.)$ around zero. \\

As a result, we have reduced the inequality \eqref{eq::inequality} to the following form:
\begin{equation*}
 F_{s}(u) \geq F_{s}(v) \hspace{2mm} \hbox{where} \hspace{2mm} u \geq v.
\end{equation*}
But, we know this is true by Lemma \ref{lem::F_increasing} in the Appendix. This proves inequality \eqref{eq::inequality} holds as claimed.
\end{proof}

\section*{Appendix}
\begin{lemma}~\cite{notes1}\label{lem:basic}
Given a channel $W$ and $\rho\geq0$, there exist a random variable $Z$ taking values in the $[0, 1]$ interval such that
\begin{equation}\label{eq:Eo}
 E_0( \rho, W) = -\log{\Expt\left[ g(\rho, Z)\right] }   
\end{equation}
where
\begin{equation}\label{eq:g}
 g(\rho, z) =  \left(  \frac{1}{2}\left( 1 + z\right)^{\frac{1}{1+\rho}} + \frac{1}{2}\left(1 - z\right)^{\frac{1}{1+\rho}} \right)^{1+\rho}.
\end{equation}
\end{lemma}
\begin{proof}
Recall
$\displaystyle
   E_0( \rho, W) = -\log\displaystyle\sum_{y} \bigg[ \hspace{2mm} \frac{1}{2} W(y \mid 0)^{\frac{1}{1+\rho}} + \frac{1}{2} W(y \mid 1)^{\frac{1}{1+\rho}} \hspace{2mm} \bigg]^{1+\rho} \nonumber 
$.
Define
\begin{equation}\label{eq:dist}
 q(y)=\frac{W(y\mid0) + W(y\mid1)}{2} \quad\text{and}\quad \Delta(y) = \frac{W(y\mid0)-W(y\mid1)}{W(y\mid0)+W(y\mid1)} \\
\end{equation}
so that
 $W(y\mid0)= q(y)[1+\Delta(y)]$ and $W(y\mid1) = q(y)[1-\Delta(y)]$.
Then, one can define the random variable $Z =  \lvert\Delta(Y)\rvert\in [0, 1]$ where $Y$ has the probability distribution $q(y)$, and obtain 
\eqref{eq:Eo} by simple manipulations.
\end{proof}

\begin{lemma}\label{lem:minus}
Given a channel $W$ and $\rho\geq0$, let $Z_1$ and $Z_2$ be independent copies of the random variable $Z$ defined in Lemma~\ref{lem:basic}.  Then,
\begin{equation}\label{eq:Eominus}
 E_0( \rho, W^{-}) = -\log{\Expt\left[ g(\rho, Z_1 Z_2)\right] }   
\end{equation}
where $g(\rho, z)$ is given by \eqref{eq:g}. 
\end{lemma}
\begin{proof}
From the definition of channel $W^{-}$ in \eqref{align:trans1}, we can write
\begin{align}
 E_0( \rho, W^{-}) = -\log\displaystyle\sum_{y_1, y_2} &\bigg[ \hspace{2mm} \frac{1}{2} W^{-}(y_1, y_2 \mid 0)^{\frac{1}{1+\rho}} + \frac{1}{2} W^{-}(y_1, y_2 \mid 1)^{\frac{1}{1+\rho}} \hspace{2mm} \bigg]^{1+\rho} \nonumber \\
= -\log\displaystyle\sum_{y_1, y_2} &\bigg[ \hspace{2mm} \frac{1}{2} \left( \frac{1}{2} W(y_1 \mid 0) W(y_2 \mid 0) +  \frac{1}{2} W(y_1 \mid 1) W(y_2 \mid 1) \right)^{\frac{1}{1+\rho}} \nonumber \\
&+ \frac{1}{2} \left( \frac{1}{2} W(y_1 \mid 1) W(y_2 \mid 0) +  \frac{1}{2} W(y_1 \mid 0) W(y_2 \mid 1) \right)^{\frac{1}{1+\rho}} \bigg]^{1+\rho} \nonumber \\
= -\log \displaystyle\sum_{y_1 y_2} &\left[ \hspace{2mm} \frac{1}{2} \left(\frac{1}{2}\right)^{\frac{1}{1+\rho}} q\displaystyle \left(y_1\right)^{\frac{1}{1+\rho}} q\left(y_2\right)^{\frac{1}{1+\rho}} \right. \nonumber \\
& \hspace{4mm} \bigl( \left( 1 + \Delta\left(y_1\right)\right)  \left( 1 + \Delta\left(y_2\right)\right)  + \left( 1 - \Delta\left(y_1\right)\right)  \left( 1 - \Delta\left(y_2\right)\right) \bigr) ^{\frac{1}{1+\rho}} \nonumber \\
& + \displaystyle  \bigl( \left( 1 - \Delta\left(y_1\right)\right)  \left( 1 + \Delta\left(y_2\right)\right) + \left( 1 + \Delta\left(y_1\right)\right)  \left( 1 - \Delta\left(y_2\right)\right)\bigr)^{\frac{1}{1+\rho}} \hspace{2mm} \bigg]^{1+\rho} \nonumber  \\
= -\log \displaystyle\sum_{y_1 y_2} &\hspace{2mm} q(y_1) \hspace{2mm} q(y_2) \hspace{2mm} \left[\frac{1}{2} \bigl(1 + \Delta(y_1) \Delta(y_2)\bigr)^{\frac{1}{1+\rho}} + \frac{1}{2} \bigl(1-\Delta(y_1)\Delta(y_2)\bigr)^{\frac{1}{1+\rho}}\right]^{1+\rho} \nonumber
\end{align}
where we used \eqref{eq:dist}. We can now define $Z_1 = \lvert\Delta(Y_1)\rvert$ and $Z_2 = \lvert\Delta(Y_2)\rvert$ 
where $Y_1$ and $Y_2$ are independent random variables with distribution $q$.
From this construction, the lemma follows. 
\end{proof}

\begin{lemma}\label{lem:plus}
Given a channel $W$ and $\rho\geq0$, let $Z_1$ and $Z_2$ be as in Lemma~\ref{lem:minus}.  Then,
\begin{equation}\label{eq:Eoplus}
E_0( \rho, W^{+}) = -\log{\Expt{\left[ \hspace{2mm} \frac{1}{2} \bigl(1 + Z_1 Z_2\bigr)\hspace{2mm} g\Bigl( \rho, \frac{Z_1 + Z_2}{1 + Z_1 Z_2}\Bigr) +\frac{1}{2} \bigl(1 - Z_1 Z_2\bigr)\hspace{2mm} g\Bigl( \rho, \frac{Z_1 - Z_2}{1 - Z_1 Z_2}\Bigr) \hspace{2mm} \right]}} 
\end{equation}
where $g(\rho, z)$ is given by \eqref{eq:g}. 
\end{lemma}
\begin{proof}
From the definition of channel $W^{+}$ in \eqref{align:trans2}, we can write
\begin{align}
E_0( \rho, W^{+}) \nonumber \\
= -\log\displaystyle\sum_{y_1, y_2, u} &\bigg[ \hspace{2mm} \frac{1}{2} W^{+}(y_1, y_2, u \mid 0)^{\frac{1}{1+\rho}} + \frac{1}{2} W^{+}(y_1, y_2, u \mid 1)^{\frac{1}{1+\rho}} \hspace{2mm} \bigg]^{1+\rho} \nonumber \\
= -\log\displaystyle\sum_{y_1, y_2, u} &\bigg[ \hspace{2mm} \frac{1}{2} \left( \frac{1}{2} W(y_1 \mid u) W(y_2 \mid 0) \right)^{\frac{1}{1+\rho}} + \frac{1}{2} \left( \frac{1}{2} W(y_1 \mid u\oplus 1) W(y_2 \mid 1)\right)^{\frac{1}{1+\rho}} \bigg]^{1+\rho} \nonumber \\
= -\log\displaystyle\sum_{y_1, y_2} &\left(  \hspace{2mm} \bigg[ \hspace{2mm} \frac{1}{2} \left( \frac{1}{2} W(y_1 \mid 0) W(y_2 \mid 0)\right)^{\frac{1}{1+\rho}}  +  \frac{1}{2} \left( \frac{1}{2} W(y_1 \mid 1) W(y_2 \mid 1) \right)^{\frac{1}{1+\rho}}  \hspace{2mm} \bigg]^{1+\rho}\right.  \nonumber \\
&+ \left. \bigg[ \hspace{2mm} \frac{1}{2} \left( \frac{1}{2} W(y_1 \mid 1) W(y_2 \mid 0)\right)^{\frac{1}{1+\rho}}  +  \frac{1}{2} \left( \frac{1}{2} W(y_1 \mid 0) W(y_2 \mid 1) \right)^{\frac{1}{1+\rho}}  \hspace{2mm} \bigg]^{1+\rho} \right)  \nonumber 
\end{align}
Using \eqref{eq:dist}, we have
\begin{align}
&E_0(\rho, W^{+}) \nonumber \\
= &-\log \displaystyle\sum_{y_1 y_2} \hspace{2mm} \frac{1}{2} \hspace{2mm} q(y_1) \hspace{2mm} q(y_2) \nonumber \\
&\hspace{16mm} \left( \hspace{2mm} \bigg[ \hspace{2mm} \bigl(\left( 1 + \Delta(y_1)\right) \left( 1 + \Delta(y_2)\right)  \bigr)^{\frac{1}{1+\rho}} +  \bigl(\left( 1 - \Delta(y_1)\right) \left( 1 - \Delta(y_2)\right)  \bigr)^{\frac{1}{1+\rho}} \hspace{2mm} \bigg]^{1+\rho} \right.  \nonumber \\
& \hspace{16mm} + \left. \bigg[ \hspace{2mm} \bigl(\left( 1 - \Delta(y_1)\right) \left( 1 + \Delta(y_2)\right)  \bigr)^{\frac{1}{1+\rho}} +  \bigl(\left( 1 - \Delta(y_1)\right) \left( 1 + \Delta(y_2)\right)  \bigr)^{\frac{1}{1+\rho}} \hspace{2mm} \bigg]^{1+\rho} \right) \nonumber \\ 
= &-\log \left( \hspace{2mm} \displaystyle\sum_{y_1 y_2} \hspace{2mm} \frac{1}{2} \hspace{2mm} q(y_1) \hspace{2mm} q(y_2) \hspace{2mm} \bigl( 1 + \Delta(y_1) \Delta(y_2)\bigr) \right.\nonumber\\
&\hspace{22mm} \bigg[ \hspace{2mm} \frac{1}{2} \left(1 + \frac{\Delta(y_1) + \Delta(y_2)}{1 + \Delta(y_1)\Delta(y_2)}\right)^{\frac{1}{1+\rho}} + \frac{1}{2} \left(1 - \frac{\Delta(y_1) + \Delta(y_2)}{1 + \Delta(y_1)\Delta(y_2)}\right)^{\frac{1}{1+\rho}}\bigg]^{1+\rho}   \nonumber \\
&\hspace{10mm} + \hspace{2mm} \displaystyle\sum_{y_1 y_2} \hspace{2mm} \frac{1}{2} \hspace{2mm} q(y_1) \hspace{2mm} q(y_2) \hspace{2mm} \bigl( 1 - \Delta(y_1) \Delta(y_2)\bigr) \nonumber \\ 
&\hspace{22mm} \left. \bigg[ \hspace{2mm} \frac{1}{2} \left(1 + \frac{\Delta(y_1) - \Delta(y_2)}{1 - \Delta(y_1)\Delta(y_2)}\right)^{\frac{1}{1+\rho}} + \frac{1}{2} \left(1 - \frac{\Delta(y_1) - \Delta(y_2)}{1 - \Delta(y_1)\Delta(y_2)}\right)^{\frac{1}{1+\rho}}\bigg]^{1+\rho}  \right)   \nonumber \\
&= -\log \left( \hspace{2mm} \displaystyle\sum_{y_1 y_2} \hspace{2mm} \frac{1}{2} \hspace{2mm} q(y_1) \hspace{2mm} q(y_2) \hspace{2mm} \bigl( 1 + \Delta(y_1) \Delta(y_2)\bigr) \hspace{2mm} g\Bigl(\rho, \frac{\Delta(y_1) + \Delta(y_2)}{1 + \Delta(y_1) \Delta(y_2)}\Bigr) \right.\nonumber \\ 
&\hspace{14mm} + \hspace{2mm} \displaystyle\sum_{y_1 y_2} \hspace{2mm} \left. \frac{1}{2} \hspace{2mm} q(y_1) \hspace{2mm} q(y_2) \hspace{2mm} \bigl( 1 - \Delta(y_1) \Delta(y_2)\bigr) \hspace{2mm}  g\Bigl(\rho, \frac{\Delta(y_1) - \Delta(y_2)}{1 - \Delta(y_1) \Delta(y_2)}\Bigr) \right)     \nonumber 
\end{align}
where $g(\rho, z)$ is defined in \eqref{eq:g}.\\\\
Similar to the $E_0(\rho, W^{-})$ case, we define $Z_1 = \lvert\Delta(Y_1)\rvert$ and 
$Z_2 = \lvert\Delta(Y_2)\rvert$ where $Y_1$ and $Y_2$ are independent random variables with distribution $q$.
 However, we should check whether this construction is equivalent to the above equation. We note that $\Delta(y) \in [-1, 1]$. 
When $\Delta(y_1)$ and $\Delta(y_2)$ are of the same sign, we can easily see (noting that $g(\rho,z)$ is symmetric about $z=0$) that
\begin{align}
\bigl( 1 + \Delta(y_1) \Delta(y_2)\bigr)\hspace{2mm} g\Bigl( \rho, \frac{\Delta(y_1) + \Delta(y_2)}{1 + \Delta(y_1) \Delta(y_2)}\Bigr)  &=  \bigl( 1 + Z_1 Z_2\bigr)\hspace{2mm}  g\Bigl(\rho, \frac{Z_1 + Z_2}{1 + Z_1 Z_2}\Bigr)  \nonumber \\ 
\bigl( 1 - \Delta(y_1) \Delta(y_2)\bigr)\hspace{2mm} g\Bigl( \rho, \frac{\Delta(y_1) - \Delta(y_2)}{1 - \Delta(y_1) \Delta(y_2)}\Bigr)  &=  \bigl( 1 - Z_1 Z_2\bigr)\hspace{2mm}  g\Bigl(\rho, \frac{Z_1 - Z_2}{1 - Z_1 Z_2}\Bigr) \nonumber 
\end{align}
When $\Delta(y_1)$ and $\Delta(y_2)$ are of the opposite sign, we note that
\begin{align}
\bigl( 1 + \Delta(y_1) \Delta(y_2)\bigr)\hspace{2mm}  g\Bigl(\rho, \frac{\Delta(y_1) + \Delta(y_2)}{1 + \Delta(y_1) \Delta(y_2)}\Bigr) =  \bigl( 1 - Z_1 Z_2\bigr)\hspace{2mm}  g\Bigl(\rho, \frac{Z_1 - Z_2}{1 - Z_1 Z_2}\Bigr) \nonumber \\ 
\bigl( 1 - \Delta(y_1) \Delta(y_2)\bigr)\hspace{2mm}  g\Bigl(\rho, \frac{\Delta(y_1) - \Delta(y_2)}{1 - \Delta(y_1) \Delta(y_2)}\Bigr) =  \bigl( 1 + Z_1 Z_2\bigr)\hspace{2mm}  g\Bigl(\rho, \frac{Z_1 + Z_2}{1 + Z_1 Z_2}\Bigr) \nonumber 
\end{align}
Since we are interested in the sum of the above two parts, we can see that the construction we propose is still equivalent. 
This concludes the proof.
\end{proof}

\begin{lemma}\label{lem::F_increasing}
For $s\in[0, 1]$, define the function $F_{s}:[0, 1]\to [1, 2^{\frac{1-s}{s}}]$ as
\begin{equation}\label{eq::F_def}
 F_{s}(x) = \displaystyle\frac{\left(1 + x^s\right)^{\frac{1}{s}}}{1+x}.
\end{equation}
Then, $ F_{s}$ is a non-decreasing function.
\end{lemma}
\begin{proof}
Taking the derivative of $F_s(x)$ with respect to $x$, we have
\begin{equation*}
 \displaystyle\frac{\partial}{\partial x} F_{s}(x) = \displaystyle\frac{(1 + x^s)^{\frac{1}{s}-1} (x^s - x)}{x (1 + x)^2} \geq 0
\end{equation*}
since $(x^s - x) \geq 0$ for $x, s\in[0, 1]$.
\end{proof}

\begin{lemma}\label{lem::f_decreasing}
For $s\in[0, 1]$, define the function $f_{s}:[0, \infty) \to [2^{\frac{s-1}{s}}, 1]$ as
\begin{equation}
 f_{s}(k) = \frac{\cosh(k s)^{\frac{1}{s}}}{\cosh(k)}.
\end{equation}
Then, $f_{s}$ is a non-increasing function. Moreover, this implies the function $g(\rho, z)$ defined in \eqref{eq:g} is non-increasing in the variable $z\in[0, 1]$ for any fixed $\rho\geq 0$. 
\end{lemma}
\begin{proof}
We can equivalently show that $\log(f_{s}(k))$ is non-increasing in $k$. Taking the first derivative gives
\begin{equation*}
 \displaystyle\frac{\partial}{\partial k} \left(\frac{1}{s}\log(\cosh(k s)) - \log(\cosh(k)) \right) = \tanh(s k) - \tanh(k) \leq 0
\end{equation*}
as $\tanh(\cdot)$ is increasing in its argument. \\
To prove the second monotonicity relation, we let $k =  \operatorname{arctanh}{z}$, and $s = \frac{1}{1+\rho}$. Then,
\begin{equation*}
g(\rho, u) \triangleq f_{\frac{1}{1+\rho}}(\operatorname{arctanh}{z}). 
\end{equation*}
Since $k =  \operatorname{arctanh}{z}$ is a monotone increasing transformation, it follows that the function $g(\rho, z)$ is non-increasing in $z$ for fixed values of $\rho$.
\end{proof}

\begin{lemma}\label{lem::h_decreasing}
The function $h:[0, \infty)\times[0, 1]\times[0, 1]\to[2^{-\rho}, 1]$ defined as
\begin{equation*}
 h(\rho, z_1, z_2)=\frac{1}{2}(1 + z_1z_2)g\left(\rho, \frac{z_1 + z_2}{1 + z_1z_2}\right) + \frac{1}{2}(1 - z_1z_2)g\left(\rho, \frac{z_1 - z_2}{1 - z_1z_2}\right) 
\end{equation*}
where $g(\rho, z)$ is given by \eqref{eq:g}, is non-increasing in the variables $z_1$ and $z_2$ separately for any $\rho\geq 0$.
\end{lemma} 
\begin{proof}
By the symmetry of $h$ with respect to $z_1$ and $z_2$, it suffices to show the claim for $z_1$ alone.  In the
expression below, we will suppress $\rho$ in all function arguments, and denote $g'(u)=\frac{\partial}{\partial u}g(\rho,u)$.
Taking the derivative of $h$ with respect to $z_1$, we get
 \begin{align*}
 \frac{\partial}{\partial z_1} h(z_1, z_2) &= \frac{1}{2}z_2\hspace{2mm} g\left( \frac{z_1 + z_2}{1 + z_1 z_2}\right) + \frac{1 - {z_2}^{2}}{2(1 + z_1 z_2)}\hspace{2mm} g'\left( \frac{z_1 + z_2}{1 + z_1 z_2}\right) \\
&- \frac{1}{2}z_2\hspace{2mm} g\left( \frac{z_1 - z_2}{1 - z_1 z_2}\right) + \frac{1 - {z_2}^{2}}{2(1 - z_1 z_2)}\hspace{2mm} g'\left( \frac{z_1 - z_2}{1 - z_1 z_2}\right)\\
 &= \frac{1}{2}z_2\Bigl[ g\left( \frac{z_1 + z_2}{1 + z_1 z_2}\right) - g\left( \frac{z_1 - z_2}{1 - z_1 z_2}\right)\Bigr]\\
&\quad+ \frac{1 - {z_2}^{2}}{2(1 + z_1 z_2)}\hspace{2mm} g'\left( \frac{z_1 + z_2}{1 + z_1 z_2}\right) \\
&\quad+ \frac{1 - {z_2}^{2}}{2(1 - z_1 z_2)}\hspace{2mm} g'\left( \frac{z_1 - z_2}{1 - z_1 z_2}\right).
\end{align*}
The last two terms that contain $g'(\cdot)$ are negative by Lemma \ref{lem::f_decreasing}, so it suffices to show that
$$
g\left(\frac{z_1 + z_2}{1 + z_1 z_2}\right) \leq g\left(\frac{z_1 - z_2}{1 - z_1 z_2}\right). 
$$
To that end, observe that, for any $z_1,z_2\in [0,1]$ we have
\begin{equation*}
 \frac{z_1 + z_2}{1 + z_1 z_2} \geq \frac{|z_1 - z_2|}{1 - z_1 z_2}
\end{equation*}
and by Lemma~\ref{lem::f_decreasing} and the symmetry of $g$ around $z=0$, the required inequality follows.
\end{proof}

\end{document}